\documentclass{article}

\title{Bounding the Dimension of Points on a Line}
\author{
	Neil Lutz\footnote{Research supported in part by National Science Foundation Grant 1445755.}\\
	Department of Computer Science, Rutgers University\\
	Piscataway, NJ 08854, USA\\
	\texttt{njlutz@rutgers.edu}
	\and
	D. M. Stull\footnote{Research supported in part by National Science Foundation Grants 1247051 and 1545028.}\\
	Department of Computer Science, Iowa State University\\
	Ames, IA 50011, USA\\
	\texttt{dstull@iastate.edu}
}

\usepackage{amsmath}
\usepackage{amssymb}
\usepackage{mdframed}
\usepackage{amsthm}
\usepackage{mathtools}
\usepackage{enumitem}
\usepackage{placeins}
\usepackage{url}

\newtheorem{thm}{Theorem}
\newtheorem{obs}[thm]{Observation}
\newtheorem{lem}[thm]{Lemma}

\newtheorem{cor}[thm]{Corollary}

\newcommand{\R}{\mathbb{R}}
\newcommand{\Z}{\mathbb{Z}}
\newcommand{\N}{\mathbb{N}}
\newcommand{\Q}{\mathbb{Q}}

\newcommand{\ve}{\varepsilon}
\newcommand{\uhr}{{\upharpoonright}}

\begin{document}
	\maketitle
	\begin{abstract}
		We use Kolmogorov complexity methods to give a lower bound on the effective Hausdorff dimension of the point $(x,ax+b)$, given real numbers $a$, $b$, and $x$. We apply our main theorem to a problem in fractal geometry, giving an improved lower bound on the (classical) Hausdorff dimension of generalized sets of Furstenberg type.
	\end{abstract}
	\section{Introduction}\label{sec:intro}
		
		In this paper we exploit fundamental connections between fractal geometry and information theory to derive both algorithmic and classical dimension bounds in the Euclidean plane.
	
		Effective fractal dimensions, originally conceived by J. Lutz to analyze computational complexity classes~\cite{Lutz03a,Lutz03b}, quantify the \emph{density of algorithmic information} in individual infinite data objects. Although these dimensions were initially defined---and have primarily been studied---in Cantor space~\cite{DowHir10}, they have been shown to be geometrically meaningful in Euclidean spaces and general metric spaces, and their behavior in these settings has been an active area of research (e.g.,~\cite{GLMP14,Mayo12,DLMT14}).
		
		This paper focuses on the \emph{effective Hausdorff dimension}, $\dim(x)$, of individual points $x\in\R^n$, which is a potentially non-zero value that depends on the Kolmogorov complexity of increasingly precise approximations of $x$~\cite{Mayo02}. Given the pointwise nature of this quantity, it is natural to investigate the \emph{dimension spectrum} of a set $E\subseteq\R^n$, i.e., the set $\{\dim(x):x\in E\}$. Even for apparently simple sets, the structure of the dimension spectrum may not be obvious, as exemplified by a longstanding open question originally posed by J. Lutz~\cite{Mill16}: \emph{Is there a straight line $L\subseteq\R^2$ such that every point on $L$ has effective Hausdorff dimension 1?}

		J. Lutz and Weihrauch~\cite{LutWei08} have shown that the set of points in $\R^n$ with dimension less than 1 is totally disconnected, as is the set of points with dimension greater than $n-1$. Turetsky has shown that the set of points in $\R^n$ of dimension exactly 1 is connected~\cite{Ture11}, which implies that every line in $\R^2$ contains a point of dimension 1.  J. Lutz and N. Lutz have shown that almost every point on any line with random slope has dimension 2~\cite{LutLut17}, despite the surprising fact that there are lines in every direction that contain no random points~\cite{LutLut15a}. These results give insight into the dimension spectra of lines, but they also leave open the question of whether or not a line in $\R^2$ can have a singleton dimension spectrum.
		
		We resolve this question in the negative with our main theorem, a general lower bound on the dimension of points on lines in $\R^2$. Our bound depends only on the dimension of the description $(a,b)$ of the line (i.e., the ordered pair giving the line's slope and vertical intercept) and the dimension of the coordinate $x$ relative to $(a,b)$.
		\begin{thm}
			\label{thm:main}
			For all $a,b,x\in\R$,
			\[\dim(x,ax+b)\geq \dim^{a,b}(x)+\min\big\{\dim(a,b),\,\dim^{a,b}(x)\big\}\,.\]
			In particular, for almost every $x\in\R$,
			$\dim(x,ax+b)=1+\min\{\dim(a,b),1\}$.
		\end{thm}
		\noindent Since $\dim(0,b)\leq\min\{\dim(a,b),1\}$, the second statement implies that every line contains two points whose dimensions differ by at least 1, and therefore that the dimension spectrum cannot be a singleton.
		
		This theorem also implies a new result in classical geometric measure theory. It has been known for more than a decade~\cite{Hitc05} that for certain classes of sets,
		\begin{equation}\label{eq:correspondence}
			\sup_{x\in E}\dim(x)=\dim_H(E)\,,
		\end{equation}
		where $\dim_H(E)$ is the (classical) Hausdorff dimension of $E$, i.e., the most standard notion of fractal dimension. Although (\ref{eq:correspondence}) does not hold in general, this correspondence suggested that effective dimensions might provide new techniques for dimension bounds in classical fractal geometry.
		
		A recent \emph{point-to-set principle} of J. Lutz and N. Lutz~\cite{LutLut17} reinforces that prospect by characterizing the Hausdorff dimension of arbitrary sets in terms of effective dimension. This principle shows that for every set $E\subseteq\R^n$ there is an oracle relative to which (\ref{eq:correspondence}) holds. In the same work, that principle is applied to give a new proof of an old result in fractal geometry. Namely, it gives an algorithmic information theoretic proof of Davies's 1971 theorem~\cite{Davi71} stating that every \emph{Kakeya set} in $\R^2$---i.e., every plane set that contains a unit segment in every direction--- has Hausdorff dimension 2.
		
		In this work, we apply the same point-to-set principle to derive a new result in classical fractal geometry from our main theorem. \emph{Furstenberg sets} generalize Kakeya sets in $\R^2$; instead of containing segments in every direction, they contain $\alpha$-(Hausdorff)-dimensional subsets of lines in every direction, for some parameter $\alpha\in(0,1]$. While the theorem of Davies gives the minimum Hausdorff dimension of Kakeya sets in $\R^2$, the minimum Hausdorff dimension of Furstenberg sets is an important open question; the best known lower bound is $\alpha+\max\{1/2,\alpha\}$.\footnote{According to Wolff~\cite{Wolff99}, this result is due, ``in all probability,'' to Furstenberg and Katznelson. See~\cite{Rela14} for a survey.} Molter and Rela~\cite{MolRel12} generalized this notion further by requiring $\alpha$-dimensional subsets of lines in only a $\beta$-dimensional set of directions, for some second parameter $\beta\in(0,1]$. They showed that any such set has Hausdorff dimension at least $\alpha+\max\{\beta/2,\alpha+\beta-1\}$. In Theorem~\ref{thm:furstenberg}, we give a lower bound of $\alpha+\min\{\beta,\alpha\}$, which constitutes an improvement whenever $\alpha,\beta<1$ and $\beta/2<\alpha$.

		For the sake of self-containment, we begin in Section~\ref{sec:prelim} with a short review of Kolmogorov complexity and effective Hausdorff dimension, along with some necessary technical lemmas. We discuss and prove our main theorem in Section~\ref{sec:main}, and we apply it to generalized Furstenberg sets in Section~\ref{sec:app}. We conclude with a brief comment on future directions.
		
	\section{Algorithmic Information Preliminaries}\label{sec:prelim}
	
		\subsection{Kolmogorov Complexity in Discrete Domains}
		
		The \emph{conditional Kolmogorov complexity} of $\sigma\in\{0,1\}^*$ given $\tau\in\{0,1\}^*$ is
		\[K(\sigma|\tau)=\min_{\pi\in\{0,1\}^*}\left\{\ell(\pi):U(\pi,\tau)=\sigma\right\}\,,\]
		where $U$ is a fixed universal prefix-free 
		Turing machine and $\ell(\pi)$ is the length of $\pi$. Any $\pi$ that achieves this minimum is said to \emph{testify} to, or be a \emph{witness} to, the value $K(\sigma|\tau)$. The \emph{Kolmogorov complexity} of $\sigma$ is $K(\sigma)=K(\sigma|\lambda)$, where $\lambda$ is the empty string. $K(\sigma)$ may also be called the \emph{algorithmic information content} of $\sigma$.
		An important property of Kolmogorov complexity is the \emph{symmetry of information} (attributed to Levin in~\cite{Gacs74}):
		\[K(\sigma|\tau,K(\tau))+K(\tau)=K(\tau|\sigma,K(\sigma))+K(\sigma)+O(1)\,.\]
		These definitions and this symmetry extend naturally to other discrete domains (e.g., integers, rationals, and tuples thereof) via standard binary encodings. See \cite{LiVit08,Nies09,DowHir10} for detailed discussion of these topics.

		\subsection{Kolmogorov Complexity in Euclidean Spaces}
		The above definitions may also be lifted to Euclidean spaces by introducing variable precision parameters~\cite{LutMay08,LutLut17}. 
		Let $x\in\R^m$, and let $r,s\in\N$.\footnote{As a matter of notational convenience, if we are given a nonintegral positive real as a precision parameter, we will always round up to the next integer. For example, $K_{r}(x)$ denotes $K_{\lceil r\rceil}(x)$ whenever $r\in(0,\infty)$.} For $\ve>0$, $B_{\ve}(x)$ denotes the open ball of radius $\ve$ centered on $x$.
		
		The \emph{Kolmogorov complexity of $x$ at precision $r$} is
		\[K_r(x)=\min\left\{K(p)\,:\,p\in B_{2^{-r}}(x)\cap\Q^m\right\}\,.\]
		The \emph{conditional Kolmogorov complexity of $x$ at precision $r$ given $q\in\Q^m$} is
		\[\hat{K}_r(x|q)=\min\left\{K(p|q)\,:\,p\in B_{2^{-r}}(x)\cap\Q^m\right\}\,.\]
		The \emph{conditional Kolmogorov complexity of $x$ at precision $r$ given $y\in\R^n$ at precision $s$} is
		\[K_{r,s}(x|y)=\max\big\{\hat{K}_r(x|q)\,:\,q\in B_{2^{-s}}(y)\cap\Q^n\big\}\,.\]
		We abbreviate $K_{r,r}(x|y)$ by $K_r(x|y)$.
		
		Although definitions based on $K$-minimizing rationals are better suited to computable analysis~\cite{Weih00}, it is sometimes useful to work instead with initial segments of infinite binary sequences. It has been informally observed that $K_r(x)=K(x\uhr r)+o(r)$, where $x\uhr r$ denotes the truncation of the binary expansion of each coordinate of $x$ to $r$ bits to the right of the binary point. We make this statement and its conditional analogues precise in Section~\ref{sec:trunc}, and we use those relationships elsewhere in the technical appendix.
		
		The following pair of lemmas show that the above definitions are only linearly sensitive to their precision parameters. Intuitively, making an estimate of a point slightly more precise only requires a small amount of information.
		\begin{lem}\label{lem:MD:3.9}
			\textup{(Case and J. Lutz~\cite{CasLut15})}
			There is a constant $c\in\N$ such that for all $n,r,s\in\N$ and $x\in\R^n$, 
			\[K_r(x)\leq K_{r+s}(x)\leq K_r(x)+K(r)+ns+a_s+c\,,\]
			where $a_s=K(s)+2\log(\lceil\frac12\log n\rceil+s+3)+(\lceil\frac12\log n\rceil+3)n+K(n)+2\log n$.
		\end{lem}
		\begin{lem}\label{lem:AIPKSCD:cond}
			\textup{(J. Lutz and N. Lutz~\cite{LutLut17})}
			For each $m,n\in\N$, there is a constant $c\in\N$ such that, for all $x\in\R^m$, $y\in\R^n$, $q\in\Q^n$, and $r,s,t\in\N$,
			\begin{itemize}
				\item[\textup{(i)}]$\hat{K}_{r}(x|q)\leq\hat{K}_{r+s}(x|q)\leq\hat{K}_{r}(x|q)+ms+2\log(1+s)+K(r,s)+c\,.$
				\item[\textup{(ii)}]$K_{r,t}(x|y)\geq K_{r,t+s}(x|y)\geq K_{r,t}(x|y)-ns-2\log(1+s)+K(t,s)+c\,.$
			\end{itemize}
		\end{lem}
		
		In Euclidean spaces, we have a weaker version of symmetry of information, which we will use in the proof of Lemma~\ref{lem:lines}.\footnote{Regarding asymptotic notation, we will treat dimensions of Euclidean spaces (i.e., $m$ and $n$) as constant throughout this work but make other dependencies explicit, either as subscripts or in the text.}
		\begin{lem}\label{lem:unichain}
			For every $m,n\in\N$, $x\in\R^m$, $y\in\R^n$, and $r,s\in\N$ with $r\geq s$,
			\begin{enumerate}
				\item[\textup{(i)}]$\displaystyle |K_r(x|y)+K_r(y)-K_r(x,y)|\leq O(\log r)+O(\log\log \|y\|)\,.$
				\item[\textup{(ii)}]$\displaystyle |K_{r,s}(x|x)+K_s(x)-K_r(x)|\leq O(\log r)+O(\log\log\|x\|)\,.$
			\end{enumerate}
		\end{lem}
		Statement (i) is a minor refinement of Theorem 3 of~\cite{LutLut17}, which treats $x$ and $y$ as constant and states that $K_r(x,y)=K_r(x|y)+K_r(y)+o(r)$. In fact, a precise sublinear term is implicit in earlier work by tracing back through several proofs in~\cite{LutLut17} and~\cite{CasLut15}. Our approach here is more direct and is left to Section~\ref{ssec:sym} of the technical appendix.
		\subsection{Effective Hausdorff Dimension}
		If $K_r(x)$ is the algorithmic information content of $x\in\R^n$ at precision $r$, then we may call $K_r(x)/r$ the \emph{algorithmic information density} of $x$ at precision $r$. This quantity need not converge as $r\to\infty$, but it does have finite asymptotes between $0$ and $n$, inclusive~\cite{LutMay08}. Although effective Hausdorff dimension was initially developed by J. Lutz using generalized martingales~\cite{Lutz03b}, it was later shown by Mayordomo~\cite{Mayo02} that it may be equivalently defined as the lower asymptote of the density of algorithmic information. 
		That is the characterization we use here. For more details on the history of connections between Hausdorff dimension and Kolmogorov complexity, see~\cite{DowHir10,Mayo08}.
		
		The \emph{(effective Hausdorff) dimension} of $x\in\R^n$ is
		\[\dim(x)=\liminf_{r\to\infty}\frac{K_r(x)}{r}\,.\]
		This formulation has led to the development of other information theoretic apparatus for effective dimensions, namely mutual and conditional dimensions~\cite{CasLut15,LutLut17}. We use the latter in this work, including in the restatement of our main theorem in Section~\ref{sec:main}. The \emph{conditional dimension of $x\in\R^m$ given $y\in\R^n$} is
		\[\dim(x|y)=\liminf_{r\to\infty}\frac{K_{r}(x|y)}{r}\,.\]
		\subsection{Algorithmic Information Relative to an Oracle}
		
		The above algorithmic information quantities may be defined relative to any oracle set $A\subseteq\N$. The \emph{conditional Kolmogorov complexity relative to} $A$ of $\sigma\in\{0,1\}^*$ \emph{given} $\tau\in\{0,1\}^*$ is
		\[K^A(\sigma|\tau)=\min_{\pi\in\{0,1\}^*}\{|\pi|\,:\,U^A(\pi,\tau)=\sigma\}\,,\]
		where $U$ is now a universal prefix-free oracle machine and the computation $U^A(\pi,\tau)$ is performed with oracle access to $A$. This change to the underlying Turing machine also induces a relativized version of each other algorithmic information theoretic quantity we have defined.
				
		Multiple oracle sets may be combined by simply interleaving them: given $A_1,\ldots,A_k\subseteq\N$, let $A=\bigcup_{i}\{kj-i+1:j\in A_i\}$. Then $K^{A_1,\ldots,A_k}(x)$ denotes $K^A(x)$.
		We will also consider algorithmic information relative to points in Euclidean spaces.
		For $y\in\R^n$, let $A_y\subseteq\N$ encode the interleaved binary expansions of $y$'s coordinates in some standard way. Then $K^y_r(x)$ denotes $K^{A_y}_r(x)$.
		We will make repeated use of the following relationship between conditional and relative Kolmogorov complexity and dimension.
		
		\begin{lem}\label{lem:AIPKSCD:4}
			\textup{(J. Lutz and N. Lutz~\cite{LutLut17})}
			For each $m,n\in\N$, there is a constant $c\in\N$ such that, for all $x\in\R^m$, $y\in\R^n$, and $r,t\in\N$,
			\[K_r^y(x)\leq K_{r,t}(x|y)+K(t)+c\,.\]
			In particular, $\dim^y(x)\leq\dim(x|y)$.
		\end{lem}
				
		In pursuing a dimensional lower bound, we will use the fact that high-dimensional points are very common. Relative to any oracle $A\subseteq\N$, it follows from standard counting arguments that almost every point $x\in\R^n$ has $\dim^A(x)=n$ and is furthermore \emph{Martin-L\"of random} relative to $A$, meaning there is some constant $c$ such that, for all $r\in\N$, $K^A_r(x)\geq nr-c$~\cite{DowHir10}.
		
		Finally, we note that all results in this paper hold, with unmodified proofs, relative to any given oracle. We present the unrelativized versions only to avoid notational clutter.

	\section{Bounding the Dimension of $(x,ax+b)$}\label{sec:main}
		In this section we prove Theorem~\ref{thm:main}, our main theorem. We first restate it in the form we will prove, which is slightly stronger than its statement in Section~\ref{sec:intro}. The dimension of $x$ in the first term is conditioned on---instead of relative to---$(a,b)$, and even when working relative to an arbitrary oracle $A$, the last term $\dim^{a,b}(x)$ remains unchanged.
		{\renewcommand{\thethm}{\ref{thm:main}}
			\begin{thm}\textup{(Restated)}
				For every $a,b,x\in\R$ and $A\subseteq\N$,
				\begin{align*}
					\dim^A(x,ax+b)\geq \dim^A(x|a,b)+
					\min\big\{\dim^A(a,b),\,\dim^{a,b}(x)\big\}\,.
				\end{align*}
				In particular, for almost every $x\in\R$,
				$\dim(x,ax+b)=1+\min\{\dim(a,b),1\}$.
			\end{thm}
			\addtocounter{thm}{-1}}
		To prove this theorem, we proceed in three major steps, which we first sketch at a very high level here. In Section~\ref{ssec:point}, we give sufficient conditions, at a given precision $r$, for a point $(x,ax+b)$ to have information content $K_r(x,ax+b)$ approaching $K_r(a,b,x)$. Notice that this is essentially the maximum possible value for $K_r(x,ax+b)$, since an estimate for $(a,b,x)$ has enough information to estimate $(x,ax+b)$ to similar precision. Informally, the conditions are
		\begin{itemize}
			\item[(i)] $K_r(a,b)$ is small.
			\item[(ii)] If $ux+v=ax+b$, then either $K_r(u,v)$ is large or $(u,v)$ is close to $(a,b)$.
		\end{itemize}

		We show in Lemma~\ref{lem:point} that when these conditions hold, we can algorithmically estimate $(a,b,x)$ given an estimate for $(x,ax+b)$.
		In Section~\ref{ssec:lines}, we give a lower bound, Lemma~\ref{lem:lines}, on $K_r(u,v)$ in terms of $\|(u,v)-(a,b)\|$, essentially showing that condition (ii) holds.
		Finally, we prove Theorem~\ref{thm:main} in Section~\ref{ssec:main} by showing that there is an oracle which allows $(a,b)$ to satisfy condition (i) without disrupting condition (ii) or too severely lowering $K_r(x,a,b)$.
		\subsection{Sufficient Conditions for a High-complexity Point}\label{ssec:point}
			Suppose that $x$, $a$, and $b$ satisfy conditions (i) and (ii) above. Then, given an estimate $q$ for the point $(x,ax+b)$, a machine can estimate $(a,b)$ by simply running all short programs until some output approximates a pair $(u,v)$ such that the line $L_{u,v}=\{(x,ux+v):x\in\R\}$ passes near $q$. Since $(u,v)$ was approximated by a short program, it has low information density and is therefore close to $(a,b)$ by condition (ii). We formalize this intuition in the following lemma.
			\begin{lem}\label{lem:point}
				Suppose that $a,b,x\in\R$, $r\in\N$, $\delta\in\R_+$, and $\ve,\eta\in\Q_+$ satisfy $r\geq \log(2|a|+|x|+5)+1$ and the following conditions.
				\begin{itemize}
					\item[\textup{(i)}]$K_r(a,b)\leq \left(\eta+\ve\right)r$.
					\item[\textup{(ii)}] For every $(u,v)\in B_1(a,b)$ such that $ux+v=ax+b$, \[K_{r}(u,v)\geq\left(\eta-\ve\right)r+\delta\cdot(r- t)\,,\]
					whenever $t=-\log\|(a,b)-(u,v)\|\in(0,r]$.
				\end{itemize}
				Then for every oracle set $A\subseteq\N$,
				\[K_r^A(x,ax+b)\geq K_r^A(a,b,x)-\frac{4\ve}{\delta}r-K(\ve)-K(\eta)-O_{a,b,x}(\log r)\,.\]
			\end{lem}
			\begin{proof}
				Let $a$, $b$, $x$, $r$, $\delta$, $\ve$, $\eta$, and $A$ be as described in the lemma statement.
				
				Define an oracle Turing machine $M$ that does the following given oracle $A$ and input $\pi=\pi_1\pi_2\pi_3\pi_4\pi_5$ such that $U^A(\pi_1)=(q_1,q_2)\in\Q^2$, $U(\pi_2)=h\in\Q^2$, $U(\pi_3)=s\in\N$, $U(\pi_4)=\zeta\in\Q$, and $U(\pi_5)=\iota\in\Q$.
				
				For every program $\sigma\in\{0,1\}^*$ with $\ell(\sigma)\leq (\iota+\zeta)s$, in parallel, $M$ simulates $U(\sigma)$. If one of the simulations halts with some output $(p_1,p_2)\in \Q^2\cap B_{2^{-1}}(h)$ such that $|p_1q_1+p_2-q_2|< 2^{-s}(|p_1|+|q_1|+3)$, then $M$ halts with output $(p_1,p_2,q_1)$. Let $c_M$ be a constant for the description of $M$.
				
				Now let $\pi_1$, $\pi_2$, $\pi_3$, $\pi_4$, and $\pi_5$ testify to $K^A_r(x,ax+b)$, $K_1(a,b)$, $K(r)$, $K(\ve)$, and $K(\eta)$, respectively, and let $\pi=\pi_1\pi_2\pi_3\pi_4\pi_5$.
				
				By condition (i), there is some $(\hat{p}_1,\hat{p}_2)\in B_{2^{-r}}(a,b)$ such that $K(\hat{p}_1,\hat{p}_2)\leq (\eta+\ve)r$, meaning that there is some $\hat{\sigma}\in\{0,1\}^*$ with $\ell(\hat{\sigma})\leq(\eta+\ve)r$ and $U(\hat{\sigma})=(\hat{p}_1,\hat{p}_2)$. A routine calculation (Observation~\ref{obs:linemachine}(i)) shows that  \[|\hat{p}_1q_1+\hat{p}_2-q_2|< 2^{-r}(|\hat{p}_1|+|q_1|+3)\,,\]
				for every $(q_1,q_2)\in B_{2^{-r}}(x,ax+b)$, so $M$ is guaranteed to halt on input $\pi$.
				Hence, let $(p_1,p_2,q_1)=M(\pi)$.
				Another routine calculation (Observation \ref{obs:linemachine}(ii)) shows that there is some
				\[(u,v)\in B_{2^{\gamma-r}}(p_1,p_2)\subseteq B_{2^{-1}}(p_1,p_2)\subseteq B_{2^0}(a,b)\]
				such that $ux+v=ax+b$, where $\gamma=\log(2|a|+|x|+5)$.
				
				We have $\|(p_1,p_2)-(u,v)\|<2^{\gamma-r}$ and $|q_1-x|<2^{-r}$, so
				\[(p_1,p_2,q_1)\in B_{2^{\gamma+1-r}}(u,v,x)\,.\]
				It follows that
				\begin{align*}
				K^A_{r-\gamma-1}(u,v,x)&\leq \ell(\pi_1\pi_2\pi_3\pi_4\pi_5)+c_M\\
				&\leq K^A_r(x,ax+b)+K_1(a,b)+K(r)+K(\ve)+K(\eta)+c_M\\
				&=K^A_r(x,ax+b)+K(\ve)+K(\eta)+O_{a,b}(\log r)\,.
				\end{align*}
				Rearranging and applying Lemma~\ref{lem:MD:3.9},
				\begin{equation}\label{eq:uvx}
					K^A_r(x,ax+b)\geq K^A_r(u,v,x)-K(\ve)-K(\eta)-O_{a,b,x}(\log r)\,.
				\end{equation}
				By the definition of $t$, if $t>r$ then $B_{2^{-r}}(u,v,x)\subseteq B_{2^{1-r}}(a,b,x)$, which implies $K^A_r(u,v,x)\geq K^A_{r-1}(a,b,x)$. Applying Lemma~\ref{lem:MD:3.9} gives
				\[K^A_r(u,v,x)\geq K^A_r(a,b,x)-O_{a,x}(\log r)\,.\]
				Otherwise, when $t\leq r$, we have $B_{2^{-r}}(u,v,x)\subseteq B_{2^{1-t}}(a,b,x)$, which implies $K^A_r(u,v,x)\geq K_{t-1}(a,b,x)$, so by Lemma~\ref{lem:MD:3.9},
				\begin{equation}\label{eq:abx}
				K^A_r(u,v,x)\geq K^A_r(a,b,x)-2(r-t)-O_{a,x}(\log r)\,.
				\end{equation}
				
				We now bound $r-t$. By our construction of $M$ and Lemma~\ref{lem:MD:3.9},
				\begin{align*}
				(\eta+\ve)r&\geq K(p_1,p_2)\\
				&\geq K_{r-\gamma}(u,v)\\
				&\geq K_r(u,v)-O_{a,x}(\log r)\,.
				\end{align*}
				Combining this with condition (ii) in the lemma statement and simplifying yields
				\[r-t\leq \frac{2\ve}{\delta}r+O_{a,x}(\log r)\,,\]
				which, together with (\ref{eq:uvx}) and (\ref{eq:abx}), gives the desired result.
			\end{proof}
		\subsection{Bounding the Complexity of Lines through a Point}\label{ssec:lines}
		In this section we bound the information content of any pair $(u,v)$ such that the line $L_{u,v}$ intersects $L_{a,b}$ at $x$. Intuitively, an estimate for $(u,v)$ gives significant information about $(a,b)$ whenever $L_{u,v}$ and $L_{a,b}$ are nearly coincident. On the other hand, estimates for $(a,b)$ and $(u,v)$ passing through $x$ together give an estimate of $x$ whose precision is
		greatest when $L_{a,b}$ and $L_{u,v}$ are nearly orthogonal. We make this dependence on $\|(a,b)-(u,v)\|$ precise in the following lemma.
		\begin{lem}\label{lem:lines}
			Let $a,b,x\in\R$. For all $u,v\in B_1(a,b)$ such that $u x+v=ax+b$, and for all $r\geq t:=-\log\|(a,b)-(u,v)\|$,
			\[K_{r}(u,v)\geq K_t(a,b) + K_{r-t,r}(x|a,b)-O_{a,b,x}(\log r)\,.\]
		\end{lem}
		\begin{proof}
			Fix $a,b,x\in\R$. By Lemma~\ref{lem:unichain}(i), for all $(u,v)\in B_1(a,b)$ and every $r\in\N$,
			\begin{equation}\label{eq:symmetry}
			K_r(u,v)\geq K_r(u,v|a,b)+K_r(a,b)-K_r(a,b|u,v)-O_{a,b}(\log r)\,.
			\end{equation}
			
			We bound $K_r(a,b)-K_r(a,b|u,v)$ first. Since $(u,v)\in B_{2^{-t}}(a,b)$, for every $r\geq t$ we have $B_r(u,v)\subseteq B_{2^{1-t}}(a,b)$, so
			\begin{equation*}
			K_r(a,b|u,v)\leq K_{r,t-1}(a,b|a,b)\,.
			\end{equation*}
			By Lemma~\ref{lem:unichain}(ii), then,
			\begin{align*}
			K_r(a,b)-K_r(a,b|u,v)&\geq K_r(a,b)-K_{r,t-1}(a,b|a,b)\\
			&\geq K_{t-1}(a,b)-O_{a,b}(\log r)\,.
			\end{align*}
			Lemma~\ref{lem:MD:3.9} tells us that
			\[K_{t-1}(a,b)\geq K_{t}(a,b)-O(\log t)\,.\]
			Therefore we have, for every $u,v\in B_1(a,b)$ and every $r\geq t$, 
			\begin{equation}\label{eq:Krz-Krzw}
			K_{r}(a,b)-K_r(a,b|u,v)\geq K_t(a,b)-O_{a,b}(\log r)\,.
			\end{equation}
			
			We now bound the term $K_r(u,v|a,b)$. Let $(u,v)\in\R^2$ be such that $ux+v=ax+b$. If  $t\leq r< t+|x|+2$, then $r-t=O_x(1)$, so by Lemma~\ref{lem:AIPKSCD:cond}(ii), $K_{r-t,r}(x|a,b)=O_x(1)$. In this case, $K_r(u,v|a,b)\geq K_{r-t,r}(x|a,b)-O_{a,b,x}(\log r)$
			holds trivially. Hence, assume $r\geq t+|x|+2$.
			
			Let $M$ be a Turing machine such that, whenever $q=(q_1,q_2)\in\Q^2$ and $U(\pi, q)=p=(p_1,p_2)\in\Q^2$, with $p_1\neq q_1$,
			\[M(\pi,q)=\frac{p_2-q_2}{p_1-q_1}\,.\]
			For each $q\in B_{2^{-r}}(a,b)\cap\Q^2$, let $\pi_q$ testify to $\hat{K}_r(u,v|q)$.
			Then
			\[U(\pi_q,q)\in B_{2^{-r}}(u,v)\cap\Q^2\,.\]
			It follows by a routine calculation (Observation~\ref{obs:routine}) that
			\[|M(\pi_q,q)-x|=\left|\frac{p_2-q_2}{p_1-q_1}-\frac{b-v}{a-u}\right|<2^{4+2|x|+t-r}\,.\]
			Thus, $M(\pi_q,q)\in B_{2^{4+2|x|+t-r}}(x)\cap\Q^2$. For some constant $c_M$, then,
			\begin{align*}
			\hat{K}_{r-4-2|x|-t}(x|q)&\leq \ell(\pi_q)+c_M\\
			&= \hat{K}_r(u,v|q)+c_M\,.
			\end{align*}
			Taking the maximum of each side over $q\in B_{2^{-r}}(a,b)\cap\Q^2$ and rearranging,
			\[K_{r}(u,v|a,b) \geq K_{r-4-2|x|-t,r}(x|a,b)- c_M\,.\]
			Then since Lemma~\ref{lem:AIPKSCD:cond}(ii) implies that
			\[K_{r-4-2|x|-t,r}(x|a,b)\geq K_{r-t,r}(x|a,b)-O_x(\log r)\,,\]
			we have shown, for every $(u,v)$ satisfying $ux+v=ax+b$ and every $r\geq t$,
			\begin{equation}\label{eq:Krwz}
			K_r(u,v|a,b)\geq K_{r-t,r}(x|a,b)-O_{a,b,x}(\log r)\,.
			\end{equation}
			The lemma follows immediately from (\ref{eq:symmetry}), (\ref{eq:Krz-Krzw}), and (\ref{eq:Krwz}).
		\end{proof}
		\subsection{Proof of Main Theorem}\label{ssec:main}
			To prove Theorem~\ref{thm:main}, we will show at every precision $r$ that there is an oracle relative to which the hypotheses of Lemma~\ref{lem:point} hold and $K_r(a,b,x)$ is still relatively large. These oracles will be based on the following lemma.
			\begin{lem}\label{lem:oracles}
				Let $n,r\in\N$, $z\in\R^n$, and $\eta\in\Q\cap[0,\dim(z)]$.
				Then there is an oracle $D=D(n,r,z,\eta)$ satisfying
				\begin{itemize}
					\item[\textup{(i)}] For every $t\leq r$, $K^D_t(z)=\min\{\eta r,K_t(z)\}+O(\log r)$.
					\item[\textup{(ii)}] For every $m,t\in\N$ and $y\in\R^m$,
					$K^{D}_{t,r}(y|z)=K_{t,r}(y|z)+O(\log r)$
					and
					$K_t^{z,D}(y)=K_t^z(y)+O(\log r)$.
				\end{itemize}
			\end{lem}
			The proof of this lemma, which uses standard methods, is deferred to Section~\ref{ssec:oracle} of the technical appendix. Informally, for some $s\leq r$ such that $K_s(z)$ is near $\eta r$, the oracle $D$ encodes $r$ bits of $z$ conditioned on $s$ bits of $z$. Unsurprisingly, access to this oracle lowers $K_t(z)$ to $K_s(z)$ whenever $t\geq s$ and has only a negligible effect when $t\leq s$, or when $r$ bits of $z$ are already known.
			
			{\renewcommand{\thethm}{\ref{thm:main}}
			\begin{thm}
				For every $a,b,x\in\R$ and $A\subseteq\N$,
				\begin{align*}
				\dim^A(x,ax+b)\geq \dim^A(x|a,b)+
				\min\big\{\dim^A(a,b),\,\dim^{a,b}(x)\big\}\,.
				\end{align*}
				In particular, for almost every $x\in\R$,
				$\dim(x,ax+b)=1+\min\{\dim(a,b),1\}$.
			\end{thm}
			\addtocounter{thm}{-1}}
			\begin{proof}
				Let $a,b,x\in\R$, and treat them as constant for the purposes of asymptotic notation here. Let $A\subseteq\N$,
				\[H=\Q\cap\big[0,\dim^A(a,b)\big]\cap\big[0,\dim^{a,b}(x)\big)\,,\] and $\eta\in H$. Let $\delta=\dim^{a,b}(x)-\eta>0$ and $\ve\in\Q_+$.
				For each $r\in\N$, let $D_r=D(2,r,(a,b),\eta)$, as defined in Lemma~\ref{lem:oracles}. 
				We claim that for every sufficiently large $r$, the conditions of Lemma~\ref{lem:point}, relativized to oracle $D_r$, are satisfied by these choices of $a$, $b$, $x$, $r$, $\delta$, $\ve$, $\eta$.
				
				Property (i) of Lemma~\ref{lem:oracles} guarantees that $K^{D_r}_r(a,b)\leq \eta r+O(\log r)$, so condition (i) of Lemma~\ref{lem:point} is satisfied for every sufficiently large $r$.
			
				To see that condition (ii) of Lemma~\ref{lem:point} is also satisfied, let $(u,v)\in B_1(a,b)$ such that $ax+b=ux+v$ and $t=-\log\|(a,b)-(u,v)\|\leq r$. Then by Lemma~\ref{lem:lines}, relativized to $D_r$, we have
				\[K^{D_r}_r(u,v)\geq K^{D_r}_t(a,b)+K_{r-t,r}^{D_r}(x|a,b)-O(\log r)\,.\]
				Therefore, by Lemma~\ref{lem:oracles} and Lemma~\ref{lem:AIPKSCD:4},
				\begin{align*}
				K^{D_r}_r(u,v)
				&\geq\min\{\eta r,K_t(a,b)\}+K_{r-t,r}(x|a,b)-O(\log r)\\
				&\geq\min\{\eta r,K_t(a,b)\}+K^{a,b}_{r-t}(x)-O(\log r)\\
				&\geq \min\{\eta r,\dim(a,b)t-o(t)\}+\dim^{a,b}(x)(r-t)-o(r)\\
				&\geq \min\{\eta r,\eta t-o(t)\}+(\eta+\delta)(r-t)-o(r)\\
				&= \eta t-o(t)+(\eta+\delta)(r-t)-o(r)\\
				&= \eta r+\delta\cdot(r- t)-o(r)\\
				&\geq (\eta-\ve)r+\delta\cdot(r-t)\,,
				\end{align*}
				whenever $r$ is large enough.
				
				For every sufficiently large $r$, then, the conclusion of Lemma~\ref{lem:point} applies here. Thus, for constant $a$, $b$, $\ve$, and $\eta$,
				\begin{align*}
				K^A_r(x,ax+b)&\geq K_r^{A,D_r}(x,ax+b)-O(1)
				\\&\geq K_r^{A,D_r}(a,b,x)-4\ve r/\delta-O(\log r)
				\\&=K^{A,D_r}_r(x|a,b)+K_r^{A,D_r}(a,b)-4\ve r/\delta-O(\log r)
				\\&=K^A_r(x|a,b)+\eta r-4\ve r/\delta-O(\log r)\,,
				\end{align*}
				where the last equality is due to the properties of $D_r$ guaranteed by Lemma~\ref{lem:oracles}.
				
				Dividing by $r$ and taking limits inferior,
				\begin{align*}
				\dim^A(x,ax+b)&\geq\liminf_{r\to\infty}\frac{K^A_r(x|a,b)+\eta r-4\ve r/\delta-O(\log r)}{r}\\
				&=\dim^A(x|a,b)+\eta -\frac{4\ve}{\delta}\,.
				\end{align*}
				Since this holds for every $\eta\in H$ and $\ve\in\Q_+$, we have
				\begin{equation*}\label{eq:primal}
					\dim^A(x,ax+b)\geq \dim^A(x|a,b)+\min\big\{\dim^A(a,b),\,\dim^{a,b}(x)\big\}\,.
				\end{equation*}
				
				The second part of the theorem statement follows easily, as relative to any given oracle for $(a,b)$, almost every $x\in\R$ is Martin-L\"of random and therefore has dimension 1. Applying Lemma~\ref{lem:AIPKSCD:4}, then, almost every $x\in\R$ has $\dim(x|a,b)\geq\dim^{a,b}(x)=1$.
			\end{proof}
			We can now easily answer the motivating question of whether or not there is a line in $\R^2$ on which every point has effective Hausdorff dimension 1.
			\begin{cor}\label{cor:gap}
				For every $a,b\in\R$, there exist $x,y\in\R$ such that
				\[\dim(x,ax+b)-\dim(y,ay+b)\geq 1\,.\]
				In particular, there is no line in $\R^2$ on which every point has dimension 1.
			\end{cor}
			\begin{proof}
				Theorem~\ref{thm:main} tells us that $\dim(x,ax+b)\geq1+\min\{\dim(a,b),1\}$ for almost every $x\in\R$. For $y=0$, we have $\dim(y,ay+b)=\dim(b)\leq\min\{\dim(a,b),1\}$.
			\end{proof}
			There are lines for which the inequality in Corollary~\ref{cor:gap} is strict. Consider, for example, a line through the origin whose slope $a$ is random. For every $x$ that is random relative to $a$, the point $(a,ax)$ has dimension $\dim(x)+\dim(a)=2$, but the origin itself has dimension 0.

		\section{An Application to Classical Fractal Geometry}\label{sec:app}
			\subsection{Hausdorff Dimension}
				As the name indicates, effective Hausdorff dimension was originally conceived as a constructive analogue to Hausdorff dimension, which is the most standard notion of dimension in fractal geometry. The properties and classical definition of Hausdorff dimension are beyond the scope of this paper; see~\cite{Falc85,Matt95} for discussion of those topics. Instead, we characterize it here according to a recent \emph{point-to-set principle}:
				\begin{thm}\label{thm:hausdorff}
					\textup{(J. Lutz and N. Lutz~\cite{LutLut17})}
					For every $n\in\N$ and $E\subseteq\mathbb{R}^n$, the Hausdorff dimension of $E$ is given by
					\[\dim_H(E)=\adjustlimits\min_{A\subseteq\N}\sup_{x\in E}\,\dim^A(x)\,.\]
				\end{thm}
			\subsection{Generalized Sets of Furstenberg Type}
				A \emph{set of Furstenberg type} with parameter $\alpha$ is a set $E\subseteq\R^2$ such that, for every $e\in S^1$ (the unit circle in $\R^2$), there is a line $\ell_e$ in the direction $e$ satisfying $\dim_H(E\cap\ell_e)\geq \alpha$. Finding the minimum possible dimension of such a set is an important open problem with connections to Falconer's distance set conjecture and to Kakeya sets~\cite{KatTao01,Wolff99}.
				
				Molter and Rela~\cite{MolRel12} introduced a natural generalization of Furstenberg sets, in which the set of directions may itself have fractal dimension. Formally, a set $E\subseteq\R^2$ is in the class $F_{\alpha\beta}$ if there is some set $J\subseteq S^1$ such that $\dim_H(J)\geq\beta$ and for every $e\in J$, there is a line $\ell_e$ in the direction $e$ satisfying $\dim_H(E\cap\ell_e)\geq \alpha$. They proved the following lower bound on the dimension of such sets.
				\begin{thm}\label{thm:molrel}
					\textup{(Molter and Rela~\cite{MolRel12})}
					For all $\alpha,\beta\in(0,1]$ and every set $E\in F_{\alpha\beta}$,
					\[\dim_H(E)\geq\alpha+\max\left\{\frac{\beta}{2}, \alpha+\beta-1\right\}\,.\]
				\end{thm}
				
				We now show that Theorem~\ref{thm:main} yields an improvement on this bound whenever $\alpha,\beta<1$ and $\beta/2<\alpha$.
				\begin{thm}\label{thm:furstenberg}
					For all $\alpha,\beta\in(0,1]$ and every set $E\in F_{\alpha\beta}$,
					\[\dim_H(E)\geq \alpha+\min\{\beta,\alpha\}\,.\]
				\end{thm}
				\begin{proof}
					Let $\alpha,\beta\in(0,1]$, $\ve\in(0,\beta)$, and $E\in F_{\alpha\beta}$. Using Theorem~\ref{thm:hausdorff}, let $A$ satisfy
					\[\sup_{z\in E}\dim^A(z)=\dim_H(E)\,.\]
					and $e\in S^1$ satisfy $\dim^A(e)=\beta-\ve>0$. Let $\ell_e$ be a line in direction $e$ such that $\dim_H(\ell_e\cap E)\geq\alpha$. Since $\dim(e)>0$, we know $e\not\in\{(0,1),(0,-1)\}$, so we may let $a,b\in\R$ be such that $L_{a,b}=\ell_e$. Notice that $\dim^A(a)=\dim^A(e)$ because the mapping $e\mapsto a$ is computable and bi-Lipschitz in some neighborhood of $e$. Let $S=\{x:(x,ax+b)\in E\}$, which is similar to $\ell_e\cap E$, so $\dim_H(S)\geq \alpha$ also. We now have
					\begin{align*}
						\dim_H(E)&=\sup_{z\in E}\dim^{A}(z)\\
						&\geq\sup_{z\in\ell_e\cap E}\dim^{A}(z)\\
						&=\sup_{x\in S}\dim^{A}(x,ax+b)\,.
					\end{align*}
					By Theorem~\ref{thm:main} and Lemma~\ref{lem:AIPKSCD:4}, both relativized to $A$, 
					\begin{align*}
						\sup_{x\in S}\dim^{A}(x,ax+b)
						&\geq\sup_{x\in S}\left\{\dim^{A,a,b}(x)+\min\{\dim^A(a,b),\dim^{A}(x|a,b)\}\right\}\\
						&\geq\sup_{x\in S}\left\{\dim^{A,a,b}(x)+\min\{\dim^A(a,b),\dim^{A,a,b}(x)\}\right\}\\
						&\geq\sup_{x\in S}\dim^{A,a,b}(x)+\min\left\{\dim^A(a),\sup_{x\in S}\dim^{A,a,b}(x)\right\}\,.
					\end{align*}
					Theorem~\ref{thm:hausdorff} gives
					\[\sup_{x\in S}\dim^{A,a,b}(x)\geq\dim_H(S)\geq\alpha\,,\]
					so we have shown, for every $\ve\in(0,\beta)$, that $\dim_H(E)\geq\alpha+\min\{\beta-\ve,\alpha\}$.
				\end{proof}
	\section{Conclusion}
		With Theorem~\ref{thm:main}, we have taken a significant step in understanding the structure of algorithmic information in Euclidean spaces. Progress in that direction is especially consequential in light of Theorem~\ref{thm:furstenberg}, which, aside from its direct value as a mathematical result, demonstrates conclusively that algorithmic dimensional methods can provide new insights into classical fractal geometry. This motivates further investigation of algorithmic fractal geometry in general and of effective Hausdorff dimension on lines in particular; improvements on our lower bound or extensions to higher dimensions would have implications for important questions about Furstenberg or Kakeya sets. Our results also motivate a broader search for potential applications of algorithmic dimension to problems in classical fractal geometry.
	\bibliography{BDPL}
	\clearpage
	\pagenumbering{arabic}
	\renewcommand*{\thepage}{A\arabic{page}}
	\appendix
	\section{Technical Appendix}\label{app:main}
	\renewcommand\thethm{\thesection.\arabic{thm}}
	\setcounter{thm}{0}

	\subsection{Initial segments versus $K$-optimizing rationals}\label{sec:trunc}
	In this section we formalize the relationship between $K_r(x)$ and the \emph{initial segment complexity} $K(x\uhr r)$. The three lemmas in this section are proved by standard techniques. We use these results elsewhere in the technical appendix, but not in the body of the paper.
	
	For $x=(x_1,\ldots,x_n)\in\R^n$ and $r\in\N$, let $x\uhr r=(x_1\uhr r,\ldots,x_n\uhr r)$, where each $x_i\uhr r=2^{-r} \lfloor 2^r x_i\rfloor$, the truncation of $x_i$ to $r$ bits to the right of the binary point. For $r\in(0,\infty)$, let $x\uhr r=x\uhr\lceil r\rceil$.
	\begin{lem}\label{lem:trunc}
		For every $m,n\in\N$, there is a constant $c$ such that for all $x\in\R^m$, $p\in\Q^n$, and $r\in\N$,
		\[|\hat{K}_r(x|p)-K(x\uhr r\,|\,p)|\leq K(r)+c\,.\]
	\end{lem}
	\begin{proof}	
		Let $m,n,r\in\N$, $x\in\R^m$, and $p\in\Q^n$. Observe that $x\uhr r\in B_{2^{-r}\sqrt{m}}(x)$, and therefore $K(x\uhr r\,|\,p)\geq \hat{K}_{r-\log(m)/2}(x|p)$.
		Thus, by Lemma~\ref{lem:AIPKSCD:cond}(i), there exists $c_1\in\N$ depending only on $m$ such that
		\[\hat{K}_r(x|p)\leq K(x\uhr r\,|\,p) + K(r) + c_1\,.\]
		
		For the other direction, observe that for every $q\in\Q^n\cap B_{2^{-r}}(x)$, we have $x\uhr r\in B_{2^{-r}(1+\sqrt{m})}(q)$, and that
		$B_{2^{-r}(1+\sqrt{m})}(q)$ contains at most $(2(1+\sqrt{m}))^m$ $r$-dyadic points, i.e., points in the set
		\[\mathcal{Q}_t^m=\{2^{-r}z:z\in\Z^m\}\,.\]
		Let $M$ be a Turing machine that, on input $(\pi,p')\in\{0,1\}^*\times\Q^n$, does the following. If $\pi=\pi_1\pi_2\pi_3$, with $U(\pi_1,p')=q\in\Q^m$, $U(\pi_2)=t\in\N$, and $U(\pi_3)=k\in\N$, then $M$ outputs the (lexicographically) $k$\textsuperscript{th} point in $\mathcal{Q}_r^m\cap B_{2^{-t}(1+\sqrt{m})}(q)$.
		
		Now let $\pi_q$ testify to $\hat{K}_r(x|p)$, let $\pi_r$ testify to $K(r)$, and let $q=U(\pi_q,p)$. There is some $k\leq(2(1+\sqrt{m}))^m$ such that $x\uhr r$ is the $k$\textsuperscript{th} point in $\mathcal{Q}_r^m\cap B_{2^{-r}(1+\sqrt{m})}(q)$; let $\pi_k$ testify to $K(k)$. Then $M(\pi_q\pi_r\pi_k,p)=x\uhr r$, so there is some machine constant $c_M$ for $M$ such that
		\begin{align*}
			K(x\uhr r\,|\,p)&\leq \ell(\pi_q)+\ell(\pi_r)+\ell(\pi_k)+c_M\\
			&=\hat{K}_r(x|p)+K(x)+K(k)+c_M
		\end{align*}
		It is well known (see, e.g.,~\cite{DowHir10}) that there is some constant $c_2$ such that
		\begin{align*}
			K(k)&\leq \log k+2\log\log k + c_2\\
			&\leq m\log(2(1+\sqrt{m}))+2\log (m\log(2(1+\sqrt{m})))+c_2\,.
		\end{align*}
		The above value depends only on $m$, as does $c_M$; let $c_3$ be their sum. Then
		\[K(x\uhr r\,|\,p)\leq \hat{K}_r(x|p)+K(r)+c_3\,,\]
		so $c=\max\{c_1,c_3\}$ affirms the lemma.
	\end{proof}
	Observing that there exists a constant $c_0$ such that, for all $m\in\N$ and $q^m\in\Q$, $|K(q)-K(q|0)|\leq c_0$, we also have the following.
	\begin{cor}\label{cor:trunc}
		For every $m\in\N$, there is a constant $c$ such that for every $x\in\R^m$ and $r\in\N$,
		\[|K_r(x)-K(x\uhr r)|\leq K(r)+c\,.\]
	\end{cor}
	\begin{cor}\label{cor:trunccond}
		For every $m,n\in\N$, there is a constant $c$ such that for all $x\in\R^m$, $y\in\R^n$, and $r,s\in\N$,
		\[|K_{r,s}(x|y)-K(x\uhr r\,|\,y\uhr s)|\leq K(r)+K(s)+c\,.\]
	\end{cor}
	\begin{proof}	
		Let $m,n,r,s\in\N$, $x\in\R^m$, and $y\in\R^n$. Let $p\in\Q^2\cap B_{2^{-s}}(y)$ be such that $K_{r,s}(x|y)=\hat{K}_r(x|p)$. Since $y\uhr s\in B_{2^{-s}\sqrt{n}}(y)$, we have $\hat{K}_{r}(x\,|\,y\uhr s)\geq K_{r,s-\log(n)/2}(x|y)$. Thus, by Lemma~\ref{lem:AIPKSCD:cond}(ii) there is a constant $c_1$ (depending on $n$) such that $\hat{K}_{r}(x\,|\,y\uhr s)\geq K_{r,s}(x|y)-K(s)-c_1$. Lemma~\ref{lem:trunc} tells us that there is a constant $c_2$ (depending on $m$) such that $K(x\uhr r\,|\,y\uhr s)\geq \hat{K}_{r}(x\,|\,y\uhr s)-K(r)-c_2$, so we have
		\[K_{r,s}(x|y)\leq K(x\uhr r\,|\,y\uhr s)+K(r)+K(s)+c_1+c_2\,.\]
		
		For the other direction, we use essentially the same technique as in the proof of Lemma~\ref{lem:trunc}, and we describe a Turing machine $M'$ that is very similar to the machine $M$ used above. On every input $(\pi,p')\in\{0,1\}^*\times\Q^n$ such that $\pi=\pi_1\pi_2\pi_3$, $U(\pi_1,p')=q\in\Q$, $U(\pi_2)=t\in\N$, and $U(\pi_3)=k\in\N$, $M'$ outputs $U(\pi_1,q')$, where $q'$ is the $k$\textsuperscript{th} point in $\mathcal{Q}_t^n\cap B_{2^{-t}(1+\sqrt{n})}(p')$.
		
		Much as before, let $\pi_x$ testify to $K(x\uhr r\,|\,y\uhr s)$, let $\pi_s$ testify to $K(s)$, and let $\pi_k$ testify to $K(k)$, where $y\uhr s$ is the $k$\textsuperscript{th} point in $\mathcal{Q}_s^n\cap B_{2^{-t}(1+\sqrt{n})}(p)$. Then
		\[M'(\pi_x,\pi_s,\pi_k)=U(\pi_x,y\uhr s)=x\uhr r\,,\]
		As $k\leq|\mathcal{Q}_s^n\cap B_{2^{-t}(1+\sqrt{n})}(p)|\leq(2(1+\sqrt{n}))^n$, there exist constants $c_{M'}$ and $c_k$ (depending on $n$) such that
		\begin{align*}
			K(x\uhr r\,|\,p)&\leq \ell(\pi_x)+\ell(\pi_s)+\ell(\pi_k)+c_{M'}\\
			&= K(x\uhr r\,|\,y\uhr s)+K(s)+K(k)+c_{M'}\\
			&= K(x\uhr r\,|\,y\uhr s)+K(s)+c_k+c_{M'}\,,
		\end{align*}
		Applying Lemma~\ref{lem:trunc} again, there is a constant $c_3$ (depending on $m$) such that $K(x\uhr r\,|p)\leq \hat{K}_r(x|p)+K(r)+c_3$. We conclude that
		\[K(x\uhr r\,|\,y\uhr s)\leq K(r)+K(s)+c_k+c_{M'}+c_3\,,\]
		therefore $c=\max\{c_1+c_2,c_k+c_{M'}+c_3\}$ affirms the lemma.
	\end{proof}
	\subsection{Approximate Symmetry of Information}\label{ssec:sym}
		Using the results of Section~\ref{sec:trunc}, it is straightforward to show that approximate symmetry of information holds for Kolmogorov complexity in Euclidean spaces.
	{
		\renewcommand{\thethm}{\ref{lem:unichain}}
		\begin{lem}\
			For every $m,n\in\N$, $x\in\R^m$, $y\in\R^n$, and $r,s\in\N$ with $r\geq s$,
			\begin{enumerate}
				\item[\textup{(i)}]$\displaystyle |K_r(x|y)+K_r(y)-K_r(x,y)\big|\leq O_{m,n}(\log r)+O_{n}(\log\log \|y\|)+O_{m,n}(1)\,.$
				\item[\textup{(ii)}]$\displaystyle |K_{r,s}(x|x)+K_s(x)-K_r(x)|\leq O_m(\log r)+O_m(\log\log\|x\|)+O_{m}(1)\,.$
			\end{enumerate}
		\end{lem}
		\addtocounter{thm}{-1}
	}
	\begin{proof}
		For (i), let $m,n,r\in\N$, $x\in\R^m$, and $y\in\R^n$. By Corollary~\ref{cor:trunc},
		\[|K_r(y)-K(y\uhr r)|\leq K(r)+O_n(1)\]
		and
		\[|K_r(x,y)-K((x,y)\uhr r)|\leq K(r)+O_{m,n}(1)\,.\]
		Notice that $K((x,y)\uhr r)=K(x\uhr r, y\uhr r)$. By Corollary~\ref{cor:trunccond},
		\[|K_r(x|y)-K(x\uhr r\,|\,y\uhr r)|\leq 2K(r)+O_{m,n}(1)\,.\]
		By the symmetry of information,
		\[K(x\uhr r\,|\,y\uhr r,K(y\uhr r)) + K(y\uhr r) - K(x\uhr r, y\uhr r)= O_{m,n}(1)\,.\]
		It is also true that
		\begin{align*}
			|K(x\uhr r\,|\,y\uhr r)-K(x\uhr r\,|\,y\uhr r,K(y\uhr r))|&\leq K(K(y\uhr r))+O_{m,n}(1)\\
			&\leq\log K(y\uhr r)+2\log\log K(y\uhr r)+O_{m,n}(1)\\
			&= O_n(\log r)+O_n(\log\log \|y\|)+O_{m,n}(1)\,.
		\end{align*}
		The second term is necessary because the integer part of $y$ is not included in the truncation length $r$. In sum,
		\begin{align*}
			|K_r(x|y)+K_r(y)-K_r(x,y)|&\leq 4K(r)+K(K(y\uhr r))+O_{m,n}(1)\\
			&\leq O_{m,n}(\log r)+O_n(\log\log\|y\|)+O_{m,n}(1)\,.
		\end{align*}
		
		The argument for (ii) is nearly identical; the only additional error is due to $K(x\uhr r,x\uhr s)-K(x\uhr r)\leq K(s)+O_m(1)\leq \log r+2\log\log r+O_m(1)$, as $s\leq r$.
	\end{proof}
	\subsection{Easy Geometric Observations}
	These observations are used in the proofs of Lemmas~\ref{lem:point} and~\ref{lem:lines}, respectively.
	\begin{obs}\label{obs:linemachine}
		Let $a,x,b\in\R$, $r\in\N$,  and $(q_1,q_2)\in B_{2^{-r}}(x,ax+b)$.
		\begin{itemize}
			\item[\textup{(i)}] If $(p_1,p_2)\in B_{2^{-r}}(a,b)$, then $|p_1q_1+p_2-q_2|< 2^{-r}(|p_1|+|q_1|+3)$.
			\item[\textup{(ii)}] If $|p_1q_1+p_2-q_2|\leq 2^{-r}(|p_1|+|q_1|+3)$, then there is some $(u,v)\in B_{2^{-r}(2|a|+|x|+5)}(p_1,p_2)$ such that $ax+b=ux+v$.
		\end{itemize}
	\end{obs}
	\begin{proof} Only the triangle inequality is needed. If $(p_1,p_2)\in B_{2^{-r}}(a,b)$, then
		\begin{align*}
		|p_1q_1+p_2-q_2|&\leq|p_1q_1+p_2-(ax+b)|+|ax+b-q_2|\\
		&<|p_1q_1-ax|+|b-p_2|+2^{-r}\\
		&<|p_1q_1-p_1x|+|p_1x-ax|+2^{1-r}\\
		&=|p_1|\cdot|q_1-x|+|x|\cdot|p_1-a|+2^{1-r}\\
		&\leq2^{-r}|p_1|+2^{-r}|x|+2^{1-r}\\
		&\leq2^{-r}(|p_1|+|x-q_1|+|q_1|+2)\\
		&<2^{-r}(|p_1|+|q_1|+3)\,.
		\end{align*}
		
		If $|p_1q_1+p_2-q_2|< 2^{-r}(|p_1|+|q_1|+3)$, then
		\begin{align*}
		|p_1x+p_2-(ax+b)|&\leq |p_1|\cdot|x-q_1|+|p_1q_1+p_2-(ax+b)|\\
		&\leq 2^{-r}|p_1|+|p_1q_1+p_2-q_2|+|q_2-(ax+b)|\\
		&< 2^{-r}|p_1|+2^{-r}(|p_1|+|q_1|+3)+2^{-r}\\
		&=2^{-r}(2|p_1|+|q_1|+4)\\
		&\leq 2^{-r}(2|a|+|x|+5)\,.
		\end{align*}
		so $(u,v)\in(p_1,ax+b-p_1x)$ affirms (ii).
	\end{proof}
	\begin{obs}\label{obs:routine}
		If $x\in\R$ and $a,b,p,q\in\R^2$ satisfy $(p_1,p_2)\in B_{2^{-r}}(a_1,a_2)$, $(q_1,q_2)\in B_{2^{-r}}(b_1,b_2)$, and $a_1x+a_2=b_1x+b_2$, then
		\[\left|\frac{p_2-q_2}{p_1-q_1}-\frac{a_2-b_2}{a_1-b_1}\right|<2^{4+2|x|+t-r}\,,\]
		whenever $t=-\log\|a-b\|$ and $r\geq t+|x|+2$.
	\end{obs}
	\begin{proof}
		From $a_1x+a_2=b_1x+b_1$ we have $(b_2-a_2)=(a_1-b_1)x$, so
		\begin{align*}
			2^{-t}&\leq\|(a_1,a_2)-(b_1,b_2)\|\\
			&=\sqrt{(a_1-b_1)^2(1+x^2)}\\
			&=|a_1-b_1|\sqrt{1+x^2}\\
			&\leq |a_1-b_1|2^{|x|}\,.
		\end{align*}
		Applying this fact and the triangle inequality several times,
		\begin{align*}
			&\left|\frac{p_2-q_2}{p_1-q_1}-\frac{a_2-b_2}{a_1-b_1}\right|\\
			&=\left|\frac{(a_1-b_1)(p_2-q_2)-(a_2-b_2)(p_1-q_1)}{(a_1-b_1)(p_1-q_1)}\right|\\
			&\leq\frac{|a_1-b_1|(|p_2-a_2|+|q_2-b_2|)+|a_2-b_2|(|p_1-a_1|+|q_1-b_1|)}{|a_1-b_1|(|a_1-b_1|-|a_1-p_1|-|b_1-q_1|)}\\
			&<\frac{2^{1-t}(2^{-r}+2^{-r})+2^{1-t}(2^{-r}+2^{-r})}{2^{-t-|x|}\cdot(2^{-t-|x|}-2^{-t-|x|-2}-2^{-t-|x|-2})}\\
			&=\frac{2^{3-t-r}}{2^{-2t-2|x|-1}}\\
			&=2^{4+2|x|+t-r}\,.
		\end{align*}
	\end{proof}
	\subsection{A Simple Oracle Construction}\label{ssec:oracle}
		Here we construct the oracles used in the proof of Theorem~\ref{thm:main}.
		{
			\renewcommand{\thethm}{\ref{lem:oracles}}
			\begin{lem}
				Let $n,r\in\N$, $z\in\R^n$, and $\eta\in\Q\cap[0,\dim(z)]$.
				Then there is an oracle $A=A(n,r,z,\eta)$ with the following properties.
				\begin{itemize}
					\item[\textup{(i)}] For every $t\leq r$, $K^A_t(z)=\min\{\eta r,K_t(z)\}+O(\log r)$.
					\item[\textup{(ii)}] For every $m,t\in\N$ and $y\in\R^m$,
					$K^{A}_{t,r}(y|z)=K_{t,r}(y|z)+O(\log r)$
					and
					$K_t^{z,A}(y)=K_t^z(y)+O(\log r)$.
				\end{itemize}
			\end{lem}
			\addtocounter{thm}{-1}
		}
		\noindent Our proof of this lemma uses the fact that conditional Kolmogorov complexity is essentially equivalent to Kolmogorov complexity relative to a finite oracle set.\footnote{In fact,~\cite{DowHir10} defines conditional Kolmogorov complexity in terms of a finite oracle, using a construction similar to the one described here.}
		\begin{obs}\label{obs:oraclecond}
			For every $k\in\N$ and $\tau=(\tau_1,\ldots,\tau_k)\in\{0,1\}^k$, define the oracle set
			\[C(\tau)=\left\{j\leq 2k:\tau_{\lfloor j/2\rfloor}=1\right\}\cup\{2k+1\}\subseteq\N\,.\]
			Then there is a constant $c$ such that for every $\sigma,\tau\in\{0,1\}^*$,
			\[\left|K(\sigma|\tau)-K^{C(\tau)}(\sigma)\right|\leq c\,.\]
		\end{obs}
		\begin{proof}
			Let $\pi\in\{0,1\}^*$ be such that $U(\pi,\tau)=\sigma$. Then given the oracle $C(\tau)$ and input $\pi$, a machine can discern $\tau$ from $2\ell(\tau)+2$ queries to $C(\tau)$ and use it to simulate $U(\pi,\tau)$. Let $\pi\in\{0,1\}^*$ such that $U^{C(\tau)}(\pi)=\sigma$. Likewise, given input $(\pi',\tau)$, a machine can compute any bit $C(\tau)$ queried in a simulation of $U^{C(\tau)}(\pi)$.
		\end{proof}
		\begin{proof}[Proof of Lemma~\ref{lem:oracles}]
			Let
			$s=\max\{t\leq r:K_{t-1}(z)<\eta r\}$.
			Observe that \[\eta r \leq K_s(z)\leq \eta r+K(s)+c\,.\]
			Let $\sigma$ be the lexicographically first time-minimizing witness to $K(z\uhr r\,|\,z\uhr s)$, and let $A=C(\sigma)$, as defined in Observation~\ref{obs:oraclecond}.
			
			Suppose $s\leq t\leq r$. Then applying a relativized version of Corollary~\ref{cor:trunc} and Observation~\ref{obs:oraclecond},
			\begin{align*}
				K^A_t(z)&\leq K^A_r(z)\\
				&\leq K^A(z\uhr r)+K(r)+O(1)\\
				&\leq K(z\uhr r\,|\,\sigma)+K(r)+O(1)\,.
			\end{align*}
			
			There exists a Turing machine $M_1$ that, on input $(\pi,\sigma)$, for $\pi\in\{0,1\}^*$, simulates $U(\sigma,U(\pi,\sigma))$. If $\pi$ is a witness to $K(z\uhr s\,|\,\sigma)$, then
			\[ M(\pi,\sigma)=U(\sigma,U(\pi,\sigma))=U(\sigma,z\uhr s)=z\uhr r\,.\]
			Thus, $K(z\uhr r\,|\,\sigma)\leq K(z\uhr s\,|\,\sigma) + c_{M_1}$, where $c_{M_1}$ is a constant for the description length of $M_1$. We now have
			\begin{align*}
				K^A_t(z)&\leq K(z\uhr s\,|\,\sigma)+K(r)+O(1)\\
				&\leq K(z\uhr s)+K(r)\\
				&\leq K_s(z)+2K(r)+O(1)\\
				&\leq \eta r+2K(r)+K(s)+O(1)\,.
			\end{align*}
			For the other direction, since $K_t^A(z)\geq K_s^A(z)$ whenever $t\geq s$, it is sufficient to show that $K_s^A(z)\geq\eta r$. We use Corollary~\ref{cor:trunc}, Observation~\ref{obs:oraclecond}, and the symmetry of information:
			\begin{align*}
				K_s^A(z)
				&\geq K^A(z\uhr s)-K(s)-O(1)\\
				&\geq K(z\uhr s\,|\,\sigma)-K(s)-O(1)\\
				&\geq K(z\uhr r\,|\,\sigma)-K(s)-O(1)\\
				&\geq K(z\uhr r)-K(\sigma)-K(s)-O(1)\\
				&= K(z\uhr r)-K(z\uhr r\,|\,z\uhr s)-K(s)-O(1)\\
				&\geq K(z\uhr r,z\uhr s)-K(z\uhr r\,|\,z\uhr s,K(z\uhr s))-K(K(z\uhr s))-2K(s)-O(1)\\
				&= K(z\uhr s)-K(K(z\uhr s))-2K(s)-O(1)\\
				&\geq K_s(z)-K(K(z\uhr s))-3K(s)-O(1)\\
				&=K_s(z)-O(\log r)\,.
			\end{align*}
			Since $K_s(z)\geq \eta r$, property (i) holds in this case.
			
			Now suppose instead that $t\leq s\leq r$. We again use Corollary~\ref{cor:trunc}, Observation~\ref{obs:oraclecond}, and the symmetry of information.
			\begin{align*}
				K_t^A(z)=&K^A(z\uhr t)-K(t)-O(1)\\
				=&K(z\uhr t\,|\,\sigma)-K(t)-O(1)\\
				\geq& K(z\uhr t\,|\,\sigma, K(\sigma))-K(t)-O(1)\\
				=&K(\sigma\,|\,z\uhr t,K(z\uhr t))+K(z\uhr t)-K(\sigma)-K(t)-O(1)\\
				\geq& K(\sigma\,|\,z\uhr t)-K(K(z\uhr t))+K(z\uhr t)-K(\sigma)-K(t)-O(1)\\
				\geq& K(\sigma\,|\,z\uhr s,t)-K(K(z\uhr t))+K(z\uhr t)-K(\sigma)-K(t)-O(1)\\
				\geq& K(z\uhr t)+K(\sigma\,|\,z\uhr s,K(z\uhr s))-K(\sigma)-K(K(z\uhr t))-2K(t)-O(1)\\
				=& K(z\uhr t)+K(z\uhr s\,|\,\sigma,K(\sigma))-K(z\uhr s)-K(K(z\uhr t))-2K(t)-O(1)\\
				\geq& K(z\uhr t)+K(z\uhr s\,|\,\sigma)-K(z\uhr s)-K(K(\sigma))-K(K(z\uhr t))\\&-2K(t)-O(1)\\
				\geq& K_t(z)+K^A_s(z)-K_s(z)-K(K(\sigma))-K(K(z\uhr t))\\&-3K(t)-2K(s)-O(1)\\
				=&K_t(z)+K^A_s(z)-K_s(z)-O(\log r)\,.
			\end{align*}
			As we have already shown that $K_s^A(z)-K_s(z)=O(\log r)$, we conclude that property (i) holds in this case as well.
			
			For property (ii), we again apply Corollary~\ref{cor:trunc}, relativized to $(z,A)$, and Observation~\ref{obs:oraclecond}, relativized to $z$, to see that
			\begin{align*}
				K_t^{z,A}(y)&\geq K^{z,A}(y\uhr t)-K(t)-O(1)\\
				&= K^z(y\uhr t\,|\,\sigma)-K(t)-O(1)\\
				&\geq K^z(y\uhr t)-K^z(\sigma)-K(t)-O(1)\\
				&\geq K_t^z(y)-K^z(\sigma)-2K(t)-O(1)\\
				&\geq K_t^z(y)-K(\sigma\,|\,z\uhr r)-2K(t)-O(1)\,,
			\end{align*}
			where the last inequality is due to Lemma~\ref{lem:AIPKSCD:4}.
			We argue that $K(\sigma\,|\,z\uhr r)$ is at most logarithmic in $r$.
			\begin{align*}
				K(\sigma\,|\,z\uhr r)&\leq K(\sigma,s,\ell(\sigma)\,|\,z\uhr r)+O(1)\\
				&\leq K(\sigma\,|\,s,\ell(\sigma),z\uhr r)+K(s)+K(\ell(\sigma))+O(1)\\
				&\leq K(\sigma\,|\,s,\ell(\sigma),z\uhr r)+O(\log r)\,.
			\end{align*}
			
			To see that the first term is constant, define a Turing machine $M_2$ that does the following. Given input $(j,k,x)$, $M_2$ simulates, for every $\pi\in\{0,1\}^{k}$ in parallel, $U(\pi,x\uhr j)$. It outputs the first such $\pi$ whose simulation halts with output $x$. We defined $\sigma$ in such a way that $M_2^z(s,\ell(\sigma),z\uhr r)=\sigma$,  so
			\[K(\sigma\,|\,s,\ell(\sigma),z\uhr r)\leq c_{M_2}\,,\]
			where $c_{M_2}$ is a constant for the length  of $M_2$'s description.
			We conclude that $K(\sigma\,|\,z\uhr r)=O(\log r)$, so $K_t^{z,A}(y)\geq K_t^{z}(y)-O(\log r)$.
			
			The argument for conditional complexity is essentially identical. By a relativized version of Corollary~\ref{cor:trunccond} and Observation~\ref{obs:oraclecond},
			\begin{align*}
				K_{t,r}^{A}(y|z)&\geq K^{z,A}(y\uhr t\,|\,z\uhr r)-K(t)-O(1)\\
				&= K(y\uhr t\,|\,z\uhr r,\sigma)-K(t)-O(1)\\
				&\geq K(y\uhr t\,|\,z\uhr r)-K(\sigma\,|\,z \uhr r)-K(t)-O(1)\\
				&\geq K_{t,r}(y|z)-K(\sigma\,|\,z \uhr r)-2K(t)-O(1)\\
				&\geq K_{t,r}(y|z)-K(\sigma\,|\,z \uhr r)-O(\log r)\,,
			\end{align*}
			and we have already shown that $K(\sigma\,|\,z\uhr r)= O(\log r)$.
		\end{proof}
\end{document}